\documentclass[10pt,conference]{IEEEtran}

\usepackage{amsthm,amsmath,amsfonts,braket,algorithm,graphicx, braket,balance}

\theoremstyle{definition} 
\newtheorem {theorem} {Theorem}

\newtheorem {lemma} {Lemma}

\usepackage{ifthen}
\newboolean{fullversionflag}
\setboolean{fullversionflag}{true}
\newcommand{\fullversion}[2]{\ifthenelse{\boolean{fullversionflag}}{{#1}}{{#2}}}

\newcommand{\heading}[1]{\text{ }\newline\textbf{{#1:}}}

\newcommand{\kb}[1]{\left[#1\right]}

\newcommand{\al}{\mathcal{A}}

\newcommand{\trd}[1]{\left|\left| #1 \right| \right|}

\newcommand{\st}{\text{ } : \text{ }}

\newcommand{\Hmin}{H_\infty}

\newcommand{\leakEC}{\lambda_{EC}}

\floatname{algorithm}{Protocol}

\usepackage{cite,mathtools}

\newcommand{\stn}[1]{\texttt{STN}_{#1}}

\newcommand{\protPM}{\prod^{\texttt{PM}}}
\newcommand{\protEB}{\prod^{\texttt{EB}}}
\newcommand{\protEBM}[1]{\prod^{\texttt{EB}}_{#1}}

\newcommand{\samp}[1]{\mathcal{S}_{{\texttt{#1}}}}

\newcommand{\px}{p_X}
\newcommand{\rej}{\texttt{rej}}

\newcommand{\minN}{{\widetilde{N}}}
\newcommand{\minNN}{{N_0}}
\newcommand{\minZ}{{n_0}}
\newcommand{\minX}{{m_0}}

\title{Finite Key Security of Simplified Trusted Node Networks}

\author{
\IEEEauthorblockN{Walter O. Krawec, Bing Wang, Ryan Brown}

\IEEEauthorblockA{School of Computing, University of Connecticut, Storrs, CT, USA}

\texttt{walter.krawec@uconn.edu}
}

\begin{document}
\maketitle
\begin{abstract}
  Simplified trusted nodes (STNs) are a form of trusted node for quantum key distribution (QKD) networks which do not require running a full QKD stack every instance (i.e., they do not need to run error correction and privacy amplification each session). Such systems hold the advantage that they may be implemented with weaker computational abilities, than regular TNs, while still keeping up with key generation rate demands. The downside is that noise tolerance is lower. However, to get a better understanding of their suitability in various scenarios, one requires practical, finite-key security bounds for STN networks. So far, only theoretical asymptotic bounds are known. In this work we derive a new proof of security for STN chains in the finite key setting. We also derive a novel cost function allowing us to evaluate when STNs would be beneficial from a computational cost perspective, compared with regular TN networks.
\end{abstract}

\section{Introduction}

Quantum key distribution (QKD) is a powerful quantum cryptographic mechanism allowing for the establishment of shared secret keys, secure against computationally unbounded adversaries.  This is unlike classical key distribution, where computational assumptions are always required to prove security.  In general, QKD systems work by having Alice stream qubits to Bob, while Bob measures these qubits.  From this, classical communication is performed to distill a final secret key.  For more information on general QKD, the reader is referred to \cite{QKD-survey,QKD-survey2,QKD-survey3}.

One of the main limitations of QKD is distance.  In general, the secret key rate degrades exponentially with distance between Alice and Bob, due to the increased chance of photon loss~\cite{Svelto10:Principles,Kaushal17:optical}.  
%
Quantum Networks can mitigate this issue.  Such networks consist of \emph{quantum repeaters} \cite{azuma2023quantum,briegel1998quantum,sangouard2011quantum} and/or \emph{trusted nodes} (TNs).  The former are still difficult to implement in practice, however they will lead to a general \emph{Quantum Internet} \cite{kimble2008quantum,caleffi2018quantum,wehner2018quantum}.  However,  most QKD networks today consist of trusted nodes, including most metro-area QKD networks (e.g., \cite{peev2009secoqc,chen2010metropolitan,zhang2018large,sasaki2011field}).  Of course, hybrid networks are also studied \cite{tysowski2018engineering,amer2020efficient}.

Trusted Nodes are QKD nodes, placed in a chain between Alice and Bob.  Each TN performs standard QKD with its neighbors, establishing pair-wise secret keys.  Finally, for Alice and Bob to establish a shared secret key, each TN will broadcast the parity of the secret keys it holds.  Bob will take all these parity announcements and XOR with his version of the secret key.  At this point, Alice and Bob will hold a correlated key that is secure against third-party adversaries.

One problem with TNs, from a computational standpoint, is that each TN must be equipped with a full QKD stack.  That is, whenever Alice and Bob wish to establish a secret key, each TN in the chain must perform (1) Error Correction (EC) and (2) Privacy Amplification (PA) twice (once with each neighbor).  Both processes can be 
 computationally intensive,
especially error correction, and so this may be a bottleneck in practical large-scale QKD network implementations.  Thus, to ensure high-speed key generation between Alice and Bob, each TN must be equipped with the computational resources needed to perform high-speed EC and PA.  This can increase the cost of the overall chain and places a bottleneck on the ``slowest'' TN in a chain.

One way to overcome this challenge are \emph{Simplified Trusted Nodes} (STNs), introduced in \cite{STN}.  Here, an STN does not need to perform EC and PA every time Alice and Bob want to establish a secret key.  Instead, each STN simply performs state preparations and measurements (e.g., BB84 \cite{QKD-BB84} style states and measurements), and broadcasts the parity of their raw measurement results (as opposed to the parity of the actual secret key after EC and PA are run as in a TN architecture).  This can be done quickly with minimal computational power, thus placing the overall bottleneck on Alice and Bob only.  Each STN will not be required to perform the time and computationally consuming tasks of EC and PA every single time they are used to establish a key.  Instead, they will be immediately free to perform another QKD session with the same, or alternative, users.  STNs may also have an advantage over TNs in security as pointed out in \cite{huang2022stream}; namely, even if an STN is later compromised, it only stores raw key information - to fully recover the secret key, an adversary needs both the raw key data \emph{and} the PA data sent between Alice and Bob.  However, while advantageous from a computational perspective (and potential cost and security perspective), STN chains have lower noise tolerances as shown, asymptotically, in \cite{STN}.  See Figure \ref{fig:STN-chain}.

\begin{figure}
    \centering
    \includegraphics[width=.8\linewidth]{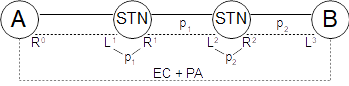}
    \caption{Showing a basic STN chain with two STNs.  Solid line: Quantum channel; Dashed line: Authenticated classical channel. Each neighboring pair will perform the quantum communication portion of BB84, establishing raw keys $R^i$ (with the right-neighbor) and $L^i$ (with the left-neighbor).  Ideally, if there is no noise, $R^i = L^{i+1}$.  Each STN will broadcast the parity of its raw keys, namely $p^i = L^i\oplus R^i$.  Bob will then take his final raw key to be the XOR these parity strings with his $L^3$ measurements; his raw key  should, in the absence of noise, now match $R^0$.  Alice and Bob then run error correction (EC) and privacy amplification (PA); STNs do not need to be involved in that final, computationally intensive, task and are instead immediately free to perform QKD with the same, or other, end-users.  Note that STNs do need to occasionally perform local QKD with their neighbors to refresh their authenticated key-pool - this is an issue we address later when comparing to a regular TN network.  Note that a regular TN network requires each neighboring pair of nodes to perform EC and PA (thus each trusted node will perform EC and PA twice) before a key is established between end users.}
    \label{fig:STN-chain}
  \end{figure}

All prior work in STN security research \cite{STN,huang2022stream,guerrini2018secure,burniston2023pre}, to our knowledge, has been restricted to asymptotic analyses.  To get a better understanding of the trade-offs when using STN networks versus TN networks, we require a finite-key security proof: that is, a bound on the number of secret key bits that can be established when sending $N$ qubits through the network (as opposed to prior work which assumed $N\rightarrow \infty$).  Such a finite key proof poses significant challenges: first we need a bound on quantum min entropy \cite{renner2008security,tomamichel2012tight}, as opposed to bounding only the von Neumann entropy \cite{QKD-Winter-Keyrate}.  Second, this bound must take into account finite key effects, along with the parity broadcasts sent by STNs.  Third, a proof must take into account that an adversary can attack all channels together, potentially gaining more information than a single channel attack. With a regular TN network, QKD is performed individually on each channel (in particular error correction and privacy amplification is run for each link), allowing one to focus on attacks on a single channel only.  Taken together, this makes a finite-key security proof a challenging problem.

In this work, we derive, for the first time to our knowledge, a finite-key security proof for an STN chain.  Our proof is general, in that it can support any number of STNs, and it assumes Eve performs any arbitrary, general, attack.  To prove this, we derive a bound on the quantum min entropy of the protocol using the quantum sampling framework of Bouman and Fehr \cite{bouman2010sampling} along with proof techniques from sampling based entropic uncertainty relations \cite{yao2022quantum} as a foundation.  However, our proof demonstrates several new techniques that may be beneficial to other researchers investigating chains of communicating nodes.

Once a finite-key rate is derived, we can begin to investigate the potential trade-offs between using STN networks and regular TN networks.  In particular, an STN chain does not need to perform EC and PA every time Alice and Bob want to establish a key (unlike TN networks).  However, they do need to perform EC and PA \emph{sometimes} in order to replenish their local key pools needed for authenticated communication channels.  Exactly how often they need to do this will depend on a variety of factors.  Considering this, is there really a cost benefit to using STNs?  Prior work is only asymptotic and could not be used to accurately answer this question in more practical finite key settings.

As a second contribution, we derive a novel cost function for STN and TN networks which takes the computational cost of EC and PA into account.  We evaluate this cost function, using our finite key bound, to provide evidence that shows STNs may be more cost effective in certain scenarios, and less cost effective in others.  In particular, in low-noise scenarios, STNs can be very cost effective; in high noise scenarios, TNs may be a preferred choice.  We comment that a similar observation was made for satellite communication using a single STN in \cite{guerrini2018secure}, though, there, communication cost was used as a metric and, furthermore, only the asymptotic scenario was considered.  Our equations will allow researchers to experiment with various parameters, including block sizes, sampling rates, and various failure parameters, to determine whether STNs are a more viable option than a standard TN network.  Indeed, a regular TN network may be more costly to implement as more expensive computational resources would be required to ``keep up'' with Alice and Bob's key generation demands.  While STNs will also need to occasionally preform EC and PA, it may be done ``in the background'' and only occasionally, with slower hardware without slowing down end-users.

\heading{Notation}
We now introduce some notation that we will use throughout this paper.  First, let $q \in \{0,1\}^N$, then for any $i = 1, \cdots, N$, we write $q_i$ to mean the $i$'th character of $q$. Let $t \subset\{1, \cdots, N\}$, then we write $q_t$ to mean the substring of $q$ indexed by $t$, namely $q = q_{t_1}q_{t_2}\cdots$.  We write $q_{-t}$ to mean the substring of $q$ indexed by the complement of $t$.  We use $wt(q)$ to mean the Hamming weight of $q$, namely the number of ones in $q$ and we use $w(q)$ to mean the relative Hamming weight of $q$, namely $w(q) = wt(q)/N$.

If $X$ is a random variable taking discrete outcomes $x_1, \cdots, x_m$, with probability $p_1, \cdots, p_m$, then we write $H(X)$ to mean the Shannon entropy of $X$, defined as $H(X) = -\sum_ip_i\log_2 p_i$.  Note that all logarithms in this paper are base two unless otherwise specified.  If $X$ is a two-outcome random variable, taking outcome $x_1$ with probability $p$ and outcome $x_2$ with probability $1-p$, then $H(X) = h(p)$, where $h(p)$ is the binary Shannon entropy function, namely $h(p) = -p\log p - (1-p)\log(1-p)$.

A quantum state, or density operator $\rho$, is a Hermitian positive semi-definite operator of unit trace, acting on some Hilbert space $\mathcal{H}$.  If $\rho_{AE}$ is a density operator acting on $\mathcal{H}_A\otimes\mathcal{H}_E$, then we write $\rho_E$ to mean the result of tracing out the $A$ system, namely $\rho_E = tr_A\rho_{AE}$.  Similarly for other, or more, systems.  To compress notation, given a pure state $\ket{\psi}$, we write $\kb{\psi}$ to mean $\ket{\psi}\bra{\psi}$.  Also, given an orthonormal basis $\mathcal{B} = \{\ket{b_0}, \cdots, \ket{b_{d-1}}\}$ and a word $q \in \{0, 1, \cdots, d-1\}^N$, we write $\ket{q}^{\mathcal{B}}$ to mean: $\ket{q_1}^{\mathcal{B}}\otimes\ket{q_N}^{\mathcal{B}} = \ket{b_{q_1}}\ket{b_{q_2}}\cdots\ket{b_{q_N}}$.  For example, if given the Hadamard $X$ basis of $X = \{\ket{+}, \ket{-}\}$, then $\ket{100}^X = \ket{-,+,+}$.  \fullversion{Finally, we define the Bell basis states as $\ket{\phi_x^y}$, for $x,y\in\{0,1\}$, as:
\begin{equation}
  \ket{\phi_x^y} = \frac{1}{\sqrt{2}}(\ket{0,x} + (-1)^y \ket{1, 1\oplus x}).
\end{equation}}{}

Let $\rho_{AE}$ be a quantum state.  Then the \emph{conditional quantum min entropy} is defined to be \cite{renner2008security}:
\begin{equation}
  \Hmin(A|E)_\rho = \sup_{\sigma_E}\max\left\{\lambda\in\mathbb{R} \st 2^{-\lambda}I_A\otimes\sigma_E - \rho_{AE} \ge 0\right\},
\end{equation}
where $I_A$ is the identity operator on the $A$ system, and where $X \ge 0$ is used to denote that operator $X$ is positive semi-definite.  The smooth conditional min entropy \cite{renner2008security} is defined to be:
$  \Hmin^\epsilon(A|E)_\rho = \sup_{\sigma_{AE}}\Hmin(A|E)_\sigma,$
where the supremum is over all density operators $\sigma_{AE}$ that are $\epsilon$-close to $\rho_{AE}$ in trace distance, namely, $\trd{\sigma_{AE} - \rho_{AE}} \le \epsilon$.  We use $\trd{X}$ to denote the trace distance of $X$.

Quantum min entropy is a vital resource in quantum cryptography as it directly relates to the amount of uniform secret randomness one may extract from a given quantum state $\rho_{AE}$, where Alice holds the $A$ system and an adversary Eve holds the $E$ system.  In particular, consider such a state, where the $A$ register is classical, consisting of $N$-bits, and the $E$ system is quantum (and possibly correlated with the $A$ system).  One may choose a random two-universal hash function $f:\{0,1\}^N \rightarrow \{0,1\}^\ell$, disclose the choice to Eve, and hash the $A$ system through $f$.  Denote the resulting state by $\sigma_{KE'}$, where $K$ is a classical register of $\ell$-bits, and $E'$ is Eve's original system combined with the choice of hash function.  Then, it was proven in \cite{renner2008security} that:
\begin{equation}\label{eq:keyrate}
  \trd{\sigma_{KE'} - I_K\otimes \sigma_{E'}} \le \sqrt{2^{-(\Hmin(A|E)_\rho - \ell)}} + 2\epsilon,
\end{equation}
where $I_K$ is a completely mixed state of $\ell$-bits.  The above process is known as \emph{privacy amplification} \cite{QKD-survey2}.  Thus, min entropy can be used to determine exactly how many secret bits $\ell$ one can extract. \fullversion{In particular if one wishes the above trace distance to be no larger than $\epsilon_{PA}$, then, one should set $\ell$ to be:
\begin{equation}
  \ell = \Hmin^\epsilon(A|E)_\rho - 2\log\frac{1}{\epsilon_{PA} - 2\epsilon}.
\end{equation}}{}

There are several very useful properties of quantum min entropy that we will use later in our proof of security.  First, given a state of the form $\rho_{AEZ} = \sum_zp(z)\kb{z}\otimes\rho_{AE}^{(z)}$, then:
\begin{equation}\label{eq:mixed-ent}
  \Hmin(A|E)_\rho \ge \Hmin(A|EZ)_\rho \ge \min_z \Hmin(A|E)_{\rho^{(z)}}.
\end{equation}
\fullversion{Thus, min entropy, conditioning on classical side information $Z$, is the ``worst-case'' entropy of each sub-event $\rho_{AE}^{(z)}$.}{}

The following lemma, proven in \cite{bouman2010sampling} (based on a lemma and proof in \cite{renner2008security}), let's us determine a bound on the min entropy of a superposition state after measuring it:
\begin{lemma}\label{lemma:superpos}
  (From \cite{bouman2010sampling}, based on \cite{renner2008security}):  Let $M$ and $N$ be two orthonormal bases of Hilbert space $\mathcal{H}_A$.  Let $\ket{\psi}_{AE} = \sum_{i\in J}\alpha_i\ket{i}^M\otimes\ket{E_i}$ be some pure quantum state.  Define the mixed state $\chi = \sum_{i\in J}\kb{i}^M\otimes\kb{E_i}$.  Then, if a measurement is made in the $N$ basis of either state (producing random variable ``$N$''), it holds that:
$    \Hmin(N|E)_\psi \ge \Hmin(N|E)_\chi - \log_2|J|.$
\end{lemma}
The above lemma states, informally, that so long as $|J|$ is ``small,'' a pure state will behave similarly to a mixed state, in terms of the entropy after measuring in an alternative basis.

We will also need the following lemma, proven in \cite{krawec2022security}:
\begin{lemma}\label{lemma:td-entropy}
  (From \cite{krawec2022security}): Let $\rho$ and $\sigma$ be two quantum states acting on the same Hilbert space such that $\frac{1}{2}\trd{\rho-\sigma}\le \epsilon$.  Let $\mathcal{F}$ be a CPTP map such that:
  \begin{align*}
    \mathcal{F}(\rho) = \sum_xp(x)\kb{x}\otimes\rho_{AE}^{(x)},\text{ and }\mathcal{F}(\sigma) &= \sum_xq(x)\kb{x}\otimes\sigma_{AE}^{(x)}
  \end{align*}
  Then, it holds that:
  \begin{equation}
    Pr\left(\Hmin^{4\epsilon + 2\epsilon^{1/3}}(A|E)_{\rho^{(x)}} \ge \Hmin (A|E)_{\sigma^{(x)}}\right) \ge 1-2\epsilon^{1/3},
  \end{equation}
  where the probability is over the random outcome $X$ in the states after mapping through $\mathcal{F}$.
\end{lemma}
The above lemma essentially allows one to bound the smooth min entropy of one state, based on the min entropy of another, assuming they are ``close enough'' in trace distance.  The bound also applies after, for example, a measurement is performed (which may be modeled as the $\mathcal{F}$ operator).

\subsection{Quantum Sampling}\label{sec:sampling}

Our proof will utilize a quantum sampling framework introduced by Bouman and Fehr in \cite{bouman2010sampling}.  For a more detailed review of this framework, the reader is referred to that original source; however, for completeness, we discuss the relevant information here.

A \emph{classical sampling strategy} over words $q \in \al^N$ (for example $\al = \{0,1\}$) is a triple $\mathcal{S}{} = (P_T, g, r)$: first a distribution $P_T$ over subsets $t \subset \{1, \cdots N\}$;  second a \emph{guess function} $g$, which outputs a real valued number based on $q_t$ for some subset $t$; and, third, a \emph{target function} $r$, which also outputs a real valued number based on $q_{-t}$.  A good sampling strategy should be one that, over the choice of random subsets according to the given distribution, the guess value evaluated on an observed portion of the word $q_t$ should closely match the target value of the unobserved portion.

\fullversion{To put this more concretely, consider the following sampling strategy from \cite{bouman2010sampling} which we denote $\samp{HW}$ which works over words in $\{0,1\}^N$.  First, the sampling strategy chooses a subset $t$ of size $m$ uniformly at random.  Next, the substring $q_t$ is observed and the guess function is simply the relative Hamming weight of $q_t$, namely $g(q_t) = w(q_t)$.  The target function is also the relative Hamming weight, namely $r(q_{-t}) = w(q_{-t})$.  One would expect that, so long as the sample size is large enough, the observed Hamming weight $g(q_t)$ should be close to the Hamming weight of the unobserved portion of the word, $r(q_{-t})$.  We will return to this example strategy later.}{}

Given a particular subset $t$, a sampling strategy induces a set of \emph{good words} which are words in $\al$ for which, assuming subset $t$ is the one that's actually chosen by the strategy, it is guaranteed that the guess and target functions will be $\delta$-close to one-another.  Formally, given a sampling strategy $\samp{}$ and subset $t$, the set of good words it induces is defined to be the set:
\begin{equation}
  \mathcal{G}^t_{\samp{}} = \{q \in \al^N \st |g(q_t) - r(q_{-t})| \le \delta\}
\end{equation}

Given these definitions, one may define the failure probability of a given sampling strategy to be:
$  \epsilon^{cl} = \max_{q\in\al^N}Pr\left(q \not\in \mathcal{G}_{\samp{}}^t\right),$
where the probability is over the subset chosen $t$, according to the sampling strategy's specification.  Note that $\epsilon^{cl}$ depends on $\delta$ also.

\fullversion{Returning to our example strategy $\samp{HW}$.}{As an example, consider a strategy denoted $\samp{HW}$, for words in $\{0,1\}^N$.  $P_T$ is the uniform distribution of sets of size $m \le N/2$, and both the guess and target functions are simply the relative Hamming weight.}  The set of good words this strategy induces is easily seen to be:
\begin{equation}\label{eq:samp-HW-good}
  \mathcal{G}^t_{HW} = \{q \in \{0,1\}^N \st |w(q_t) - w(q_{-t})|\le\delta\}.
\end{equation}
Then, in \cite{bouman2010sampling}, the following lemma was proven:
\begin{lemma}\label{lemma:samp-HW}
  (From \cite{bouman2010sampling}): Given $\samp{HW}$ as defined above, the error probability is found to be:
  \[
  \max_{q\in\{0,1\}^N}Pr\left(q\not\in\mathcal{G}^t_{HW}\right) \le 2\exp\left(-\delta^2\frac{mN}{N+2}\right) := \epsilon_{HW}^{cl}
  \]
\end{lemma}

A classical sampling strategy may be promoted to a quantum one in a natural way.  Fix an orthonormal basis of dimension $|\al|$.  We label it here as simply $\{\ket{0}, \cdots, \ket{|\al|-1}\}$, though the basis may be arbitrary.  Then, let $\ket{\psi}_{AE}$ be some quantum state where the $A$ register lives in a space of dimension $|\al|^N$ (i.e., it consists of $N$ systems, each system of dimension $|\al|$).  The $E$ register is arbitrary.  Note this system need not be separable and can, in fact, be arbitrary within this space.  Then, given some classical sampling strategy, the quantum version simply chooses a subset as before, and will measure those qudits, indexed by $t$ in the given basis to produce a classical word $q_t \in \al^{|t|}$.  The question, then, becomes what can we say about the remaining, unmeasured, systems?

Bouman and Fehr's main result is to show that, essentially, the remaining unmeasured portion must collapse to a superposition consisting of words, with respect to the given basis, that are $\delta$-close in target function, to the guess $g(q_t)$.

To define this formally, fix a sampling strategy $\samp{}$ over words in $\al^N$ and let $\mathcal{B}$ be a $|\al|$-dimensional orthonormal basis.  The sampling strategy induces a set of good words $\mathcal{G}^t$.  Consider the following subspace, denoted $\mathcal{G}^t_{\samp{},\mathcal{B}}$:
\begin{equation}
  \mathcal{G}^t_{\samp{},\mathcal{B}} = \text{span}\left\{\ket{q}^\mathcal{B} \st q \in \mathcal{G}^t\right\}\otimes\mathcal{H}_E.
\end{equation}
Then, a quantum state $\ket{\nu^t}$ is said to be an \emph{ideal state}, with respect to the given subset $t$, if $\ket{\nu^t}\in\mathcal{G}_{\samp{},\mathcal{B}}$.  Note that, given $\ket{\nu^t}$, if the sampling strategy actually chooses subset $t$ and measures those qudits indexed by $t$ in basis $\mathcal{B}$ resulting in outcome $q_t\in\al^{|t|}$, it is guaranteed that the unmeasured state must collapse to a superposition of the form:
$  \ket{\nu^t_q} = \sum_{i \in J_q}\ket{i}^\mathcal{B} \ket{E_i^{q,t}},$
where:
\begin{equation}
  J_q = \{i \in \al^{N - |t|} \st |r(i) - g(q_t)| \le \delta\}.
\end{equation}
In general, these ideal states are nice to work with as they are ``well behaved'' after a  measurement is made.  Bouman and Fehr's main result can be summarized in the theorem below:
\begin{theorem}\label{thm:sampling}
  (From results in \cite{bouman2010sampling}): Let $\samp{}$, be a classical sampling strategy over words of length $N$ in some $d$-dimensional alphabet, with error probability $\epsilon^{cl}$.  Then, given a quantum state $\ket{\psi}_{AE}$ where the $A$ register is a $d^N$ dimensional Hilbert space, and any $d$-dimensional orthonormal basis $\mathcal{B}$, there exists a collection of ideal states $\{\ket{\nu^t}\}$, indexed by possible subsets $t$, such that $\ket{\nu^t} \in \mathcal{G}^t_{\samp{}, \mathcal{B}}$ and:
$    \frac{1}{2}\trd{\sum_tP_T(t)\kb{t}\otimes\left(\kb{\psi} - \kb{\nu^t}\right)} \le \sqrt{\epsilon^{cl}}.$
\end{theorem}

Thus, on average over subset choices, the given ``real state'' $\ket{\psi}$ should be close, in trace distance, to these ideal states $\ket{\nu^t}$.  How close they are depends on the analysis of a classical sampling strategy.

\heading{Sampling Strategies}
We now introduce a sampling strategy which we will need in our proof later.  Consider the following sampling strategy which we denote by $\samp{STN}$ involving $p+2$ parties over words $q = (r^0,l^1r^1,l^2r^2,\cdots l^pr^p,l^{p+1}) \in \{0,1\}^N\times\{0,1\}^{2N}\times\cdots\times\{0,1\}^{2N}\times\{0,1\}^N = \Sigma_N$ (we consider $l^ir^i$ to be two sequential $N$-bit strings).  For notation, given a subset $t \subset\{1, \cdots, N\}$, then we write $q[t]$ to mean the following string:
  \begin{equation}
    q[t] := r^0_{t} \oplus \left(l^1_{t} \oplus r^1_{t}\right) \oplus \cdots \oplus \left(l^p_{t}\oplus r^p_{t}\right)\oplus l^{p+1}_{t}.
  \end{equation}

The sampling strategy $\samp{STN}$ then acts as follows:
(1) A subset $t$ is chosen uniformly at random such that  $t \subset\{1, 2, \cdots, N\}$ and $|t| = m < N/2$.
(2) Next, $q[t]$ is observed and the relative Hamming weight is computed.  This is used as a guess for the relative Hamming weight of the unobserved string $q[-t]$, where $-t = \{1, \cdots, N\} \setminus t$.  This implies the set of good words is:
  \begin{equation}\label{eq:good-STN}
    \mathcal{G}^{t}_{STN} := \left\{q \in \Sigma_N \st |w(q[t]) - w(q[-t])| \le \delta\right\}.
  \end{equation}

The above sampling strategy will essentially model the sampling information we will learn in the STN network we analyze later.  \fullversion{The string $r^0$ and $l^{p+1}$ will represent Alice and Bob's information respectively, while each pair $l^ir^i$ will represent the data held by the $i$'th STN.  Since STNs will simply broadcast the parity of their data (i.e., $l^i\oplus r^i$) and not the individual data (not $l^i$ and $r^i$ separately), the sampling strategy only has access to the XOR of these pair-wise strings.}{}

The failure probability of this strategy is analyzed in the following lemma:
\begin{lemma}\label{lemma:samp-stn}
  Let $\delta > 0$, and let $m \le N/2$.  Then, the failure probability of $\samp{STN}$ is upper bounded by:
  \begin{equation}
    \max_{q\in\Sigma_N}Pr\left(q \not\in \mathcal{G}^t\right) \le 2\exp\left(-\delta^2\frac{mN}{N+2}\right).
  \end{equation}
\end{lemma}
\begin{proof}
  Fix $q = (r^0, l^1r^1, \cdots, l^pr^p, l^{p+1}) \in \Sigma_N$.  Then, define a new string $\widetilde{q} = \in \{0,1\}^N$  to be $\widetilde{q} = r^0\oplus(l^1\oplus r^1) \oplus (l^2\oplus r^2)\oplus\cdots\oplus (l^p \oplus r^p) \oplus l^{p+1}$.  It is clear that, for any subset $t$, it holds that $q \not\in \mathcal{G}^t_{STN}$ implies that $\widetilde{q}\not\in\mathcal{G}^t_{HW}$, where $\mathcal{G}^t_{HW}$ is defined in Equation \ref{eq:samp-HW-good}.  Since this is true for any subset and since $q$ was arbitrary, the result follows from Lemma \ref{lemma:samp-HW}.
\end{proof}


\section{Simplified Trusted Nodes}
Simplified trusted nodes (STNs), originally introduced in \cite{STN}, act as regular trusted nodes, except they do not need to perform any sampling, error correction, or privacy amplification whenever Alice and Bob want to establish a secret key.  We consider a chain topology where Alice and Bob are connected through $p$ STNs (see Figure \ref{fig:STN-chain}).  We assume that each neighboring pair of nodes has access to a classical authenticated channel; we also assume Alice and Bob have an authenticated channel.  We do not require every possible pair of STN's to share an authenticated channel, however, only adjacent pairs in the chain. Note that, such a channel may be implemented in an information theoretic secure way using a small pre-shared key \cite{QKD-survey2} (which must later be refreshed as we discuss in our Evaluation section).  We comment on these issues more later.

We will analyze the finite-key setting of an STN chain.  Here, Alice and Bob wish to derive a secret key using $N$ rounds of communication.  Let $\stn{1}, \stn{2}, \cdots, \stn{p}$ be the $p$ STNs.  Alice will stream $N$ qubits to $\stn{1}$; these $N$ qubits will be prepared in either the $Z$ or $X$ basis.  Furthermore, the basis choice will be biased so that $X$ basis states are sent with probability $\px \le 1/2$.  $\stn{1}$ will measure the incoming qubits in either the $Z$ or $X$ basis, choosing randomly, though biased so that the $X$ basis is chosen with the same probability $\px$.  This parameter $\px$ may be optimized over by users.  In parallel, each $\stn{i}$ will stream $N$ qubits to $\stn{i+1}$ who will measure them; each party choosing the $Z$ and $X$ bases randomly (and, also, biasing the basis choice).  Finally, the last STN, $\stn{p}$, will stream $N$ qubits similarly to Bob who will measure in a random basis, similar to the STNs.

Following this, neighboring parties send their basis choices to each other and discard any of the $N$ rounds where they did not choose the same basis (both for sending and measuring on a single link).  It is expected that each neighboring party keeps $N(\px^2 + (1-\px)^2)$ of the $N$ rounds.  Of these kept rounds, parties separate their data into $Z$ rounds (where both parties chose to send/measure in the $Z$ basis) and $X$ rounds.  

Consider $\stn{i}$: it holds data shared with $\stn{i-1}$ (or Alice, if $i=1$) and also $\stn{i+1}$ (or Bob, if $i=p$).  Call the data shared with $\stn{i-1}$ the ``left'' data and $\stn{i+1}$ the ``right'' data string (each further divided into $Z$ and $X$ data strings).  Each $\stn{i}$ will send to $\stn{i+1}$ (or Bob if $i=p$), the parity (or the XOR) of that STN's data - namely, $\stn{i}$ will send $L_Z^i\oplus R_Z^i$, where $L_Z^i$ and $R_Z^i$ are the left and right data strings for the $Z$ basis, and will similarly send $L_X^i\oplus R_X^i$.  Of course, it's possible that the bit-sizes of these two strings are not identical - thus the right-most bits of the largest string are simply discarded.  $\stn{i+1}$ will receive this message and pass it along to $\stn{i+2}$ while also repeating the above for this STN's own individual left and right data strings.  Note that all this classical communication is done using the authenticated channels.  The above process repeats for all STN's until Bob finally receives the parity strings from all $p$ STN's.  

These parity strings, sent by the STNs, may all be of different sizes (though they should not differ too much in expected size), so Bob simply takes the minimum of them all, including the size of his own measurement string, and discards the right-most bits from all bit strings.  Let $\minZ$ be the size of the smallest parity string or his own bit string shared with $\stn{p}$ for the $Z$ basis data and $\minX$ be the same, but for the $X$ basis data.   He XOR's all parity strings together with his measurement data.  For the $Z$ measurement data, this will constitute his raw key; for the $X$ measurement data, this will constitute his channel test data.  Bob then sends to Alice the sizes $\minZ$, $\minX$ and also his $X$ basis data string (after XOR'ing with the STN's $X$ basis parity strings) using their authenticated channel, separate from the pair-wise authenticated channels used by the STN chain.  (Though, of course, $\minX$ may be inferred from the actual $X$ basis data string that's sent).

Alice checks the number of errors in the $X$ basis string - ideally, the $X$ basis data that Bob sent to her should match exactly the $X$ basis data she initially sent to $\stn{1}$.  Any non-matching outcome is counted as an error.  Assuming the error rate is low enough (to be determined later), Alice and Bob will next run an error correction and privacy amplification protocol on their $Z$ basis string to distill their final secret key.  Error correction and privacy amplification are standard processes in QKD; for more details, we refer the reader to \cite{QKD-survey2,QKD-survey3}.  Note that only Alice and Bob need to perform error correction and privacy amplification each time they want to establish a key - the STN's are not required for this, and are free to perform QKD again immediately with other users or the same users - they do not need to spend computational time and resources on error correction and privacy amplification each time a pair of users wants to establish a key.  The STNs will need to later refresh their authenticated channel key-pool, however this may be done infrequently and is something we consider later in our Evaluation section.


\section{Security Analysis}

We now compute the key-rate of the STN chain network discussed in the previous section. \fullversion{We will actually analyze an entanglement-based (EB) version, which we denote $\protEB$ where, instead of the prepare and measure based system where Alice and $\stn{1}$ communicate; $\stn{1}$ and $\stn{2}$ communicate, and so on (with the adversary Eve probing each link in arbitrary manners), we instead consider the case where Eve is allowed to prepare all qubits utilized by the network, entangling them arbitrarily with her ancilla, and sending the correct number of qubits to each party respectively.  We will also make additional simplifications to the protocol which can only benefit Eve.  To prove that security of this entanglement based version (which we will formally define below) will imply security of the prepare and measure version, denoted $\protPM$, discussed in the previous section, we will actually derive several intermediate protocols, building towards the final entanglement based one.  Once the entanglement-based protocol is defined, we will show how the min entropy of the system can be computed, giving us an immediate lower-bound on the key-rate of the protocol.}
{We will actually analyze an entanglement-based (EB) version, denoted $\protEB$, where Eve is allowed to prepare all qubits utilized by the network, entangling them arbitrarily with her ancilla, and sending the correct number of qubits to each party respectively.  We will also make additional simplifications to the protocol which can only benefit Eve. To construct this reduction, we will actually derive several intermediate protocols, building towards the final entanglement based one.  Once the entanglement-based protocol is defined, we will show how the min entropy of the system can be computed, giving us an immediate lower-bound on the key-rate of the protocol.
  }

\heading{Reduction to an Entanglement Based Protocol}
We will show how $\protPM$ can be simplified to an entanglement based version where (1) Eve prepares all quantum states and (2) there are no mismatches in basis measurements.  To do so, we will construct three intermediate protocols, denoted $\protEBM{0}$, $\protEBM{1}$ and $\protEBM{2}$.  From the last protocol, we will derive the final entanglement based version, denoted $\protEB$.  For each step, we will show that security of each newly derived protocol implies security of the previous.

For the first step of our reduction: we may replace the steps where a node (a node being Alice, Bob, or an STN) chooses to send one of four qubit states to its right-most neighbor with the following: A node will create a Bell pair $\ket{\phi_0^0} = \frac{1}{\sqrt{2}}(\ket{00} + \ket{11})$, and keep one qubit local, while sending the other qubit to that node's right-most neighbor.  Later, parties will choose either the $Z$ or $X$ basis to measure their respective particles in.  It is not difficult to see that this will be mathematically identical to $\protPM$.  We call this protocol $\protEBM{0}$  Next, we allow Eve to create the initial state, creating a new protocol $\protEBM{1}$:

$ $\newline (1) Let $N$ be the total number of rounds of the network used to establish a secret key (specified by Alice and Bob), and $p$ the total number of STNs in the network chain.
$ $\newline (2)\label{EBM1:state-prep} Eve prepares a quantum state $\ket{\psi}_{AT^1T^2\cdots T^p BE}$, where the $A$ and $B$ registers consist of $N$ qubits each, while each $T^i$ register consists of $2N$ qubits each.  The $A$ and $B$ registers are sent to Alice and Bob respectively, while the $T^i$ register is sent to $\stn{i}$, for $i=1, 2, \cdots, p$.
  \fullversion{
  \begin{itemize}
  \item For notation, we will divide each $2N$ qubit $T^i$ register into two, $N$-qubits registers, $L^i$ and $R^i$; that is, $T^i = L^iR^i$.  The $L^i$ register will simulate the $N$ qubits received from the node to the left of $\stn{i}$ in the $\protEBM{0}$ version of the protocol, while the $R^i$ register will simulate the stored $N$ qubits from the Bell pairs sent to the party to the right.
    \item To further simplify notation, we will also refer to the $A$ register as $R^0$ (i.e., $A = R^0$) and the $B$ register as $L^{p+1}$.  This allows us to talk about ``link $i$'' which consists of registers $R^i$ and $L^{i+1}$, for $i = 0, \cdots, p$.
  \item Ideally, if Eve is ``honest,'' the state she prepares should consist of $N$ independent Bell states on each link, unentangled with Eve's ancilla $E$. 
    Of course, Eve may prepare any state; furthermore we do not assume the state has any iid structure to it (i.e., we prove security against arbitrary, general, attacks).
  \end{itemize}}
{For notation, we divide each $T^i$ into two $N$ qubit registers: $T^i=R^iL^i$; we also denote $A$ as $R^0$ and $B$ as register $L^{p+1}$ (see also Figure \ref{fig:STN-chain}).
  }
$ $\newline (3) \label{EBM1:basis-choice}For every $i$'th link, consisting of $R^iL^{i+1}$, for $i = 0, \cdots, p$, the party to the left (the $i$'th party, with Alice being party $0$) and to the right (the $i+1$'th party, with Bob being party $p+1$), will choose strings $\Theta^i, \Psi^{i} \in \{0,1\}^{N}$ respectively such that each bit of $\Theta^i$ and $\Psi^i$ are chosen independently at random with $Pr\left(\Theta^i_j = 1\right) = Pr\left(\Psi^i_j = 1\right) = \px$ for every $j$.  $\Theta^i$ will represent the measurement basis choice for $R^i$ (with a one in index $j$ implying an $X$ basis measurement of qubit $j$, while a zero indicates a $Z$ basis measurement); $\Psi^i$ represents the same, but for $L^{i+1}$.  Note, no measurements are performed yet.
$ $\newline (4) \label{EBM1:reject} Let $\rej^i\in\{0,1\}^N$ be a string such that $\rej^i_j = 1$ if $\Theta_j^i \ne \Psi_j^i$ (and zero otherwise). This represents the string of rejected qubits (if $\rej^i_j = 1$, then qubit $j$ will be measured in opposite bases and so must be rejected).  Thus, all qubits where $\rej^i_j = 1$ are discarded from both left and right registers on each link (i.e., they are simply traced out).  Each link $i$ now consists of $N^i$ qubits, where $N^i = N - wt(\rej^j)$.
$ $\newline(5) Parties now measure the remaining $N^i$ qubits using the basis indicated in their (now matching) choice strings $\Theta^i$ and $\Psi^i$.  \fullversion{This data is split into $Z$ and $X$ measurement strings.  Let $m^i$ be the total number of $X$ basis measurements on this link and $n^i$ be the total number of $Z$ basis measurements.
$ $\newline(6) Each STN will send the parity (XOR) of their $Z$ and $X$ measurement strings to their right-most neighbor who will ultimately continue to forward the information to Bob as in $\protPM$.
$ $\newline (7) Bob will XOR the received parity strings to his respective $Z$ and $X$ measurement strings.  If (as is likely) these strings are not of equal length, he will take the smallest size and discard anything to the right of the cut off point.  He will then send his $X$ measurement results (XOR'd with the STN's parity strings) to Alice for error checking.  Ideally, her $X$ basis measurement results will match his sent value.  Alice counts the relative number of errors in this $X$ basis string and if this number (the noise) is too high (to be discussed), she aborts.  Otherwise, Alice's $Z$ basis measurement string will be used as her raw key while Bob's $Z$ basis string, XOR'd with the STN's $Z$ basis parity strings, will be used as his raw key.
$ $\newline (8) Alice and Bob run error correction and privacy amplification as normal.}{ The rest of the protocol, then, is identical to the prepare-and-measure version.}

We wish to simplify the above protocol even further.  Notice that the overall raw key size cannot exceed $\minN = \min_i N^i = N - \max_iwt(\rej^i)$ bits, due to the fact that the smallest consistent measurement results (by that, we mean, measurement results resulting from instances where neighboring parties chose the same basis) are a bottleneck of the entire chain.  Other qubits, beyond this range, are discarded in a deterministic manner.  Furthermore, the discarding of rejected systems leaves all parties with a mixed state (even before all nodes measure in their respective basis).  Thus, it would be better for Eve if parties always agreed on the correct basis choice (i.e., there were no mismatches), and, instead, Eve simply prepared a smaller, but pure, state initially.  That is, Eve will prepare a pure state where each $R^i$ and $L^{i+1}$ register holds $\minN$ qubits and each link will choose a subset $\Theta^i$, setting $\Psi^i = \Theta^i$.  Such a system can only give Eve more information than the mixed state that would result in $\protEBM{1}$ above.

\fullversion{Of course, we have the following problem: what should we set the register sizes $\minN$ to be now?}{}  In an actual run of the protocol $\protEBM{1}$, this size depends on random choices of all honest parties (Alice, Bob, and the $p$ STNs). However, importantly, Eve cannot control directly the size of $\minN$ - instead it is independent of her initial state.  Furthermore, since $\minN$ depends only on the largest $wt(\rej^i)$, we may also find a lower-bound on $\minN$ using Hoeffding's inequality, treating $\rej^j$ as a random variable where $Pr(\rej^i_j = 1) = 2\px(1-\px)$.  The expected value of $wt(\rej^i)$ is simply $2N\px(1-\px)$.

Let $\epsilon_{\text{abort}} > 0$ be given, and define $\beta$ to be:
\begin{equation}
  \beta = \sqrt{\frac{\ln\frac{2}{\epsilon_{\text{abort}}}}{2N}}.
\end{equation}
Then, by Hoeffding's inequality, we find:
\begin{equation}
  Pr\left(|N^i - N(1-2\px(1-\px))| \ge \beta N\right) \le \epsilon_{\text{abort}}
\end{equation}
Since the above is true for every link $i$, if we set $\minN = N(1-2\px(1-\px) - \beta)$, it will hold that, except with probability at most $(p+1)\epsilon_{\text{abort}}$, the size of each system, after discarding rejected rounds in $\protEBM{1}$, will be no smaller than $\minN$.  We may, therefore, adjust the above protocol so that parties abort the entire protocol if it ever holds that $N^i < \minN$.  It is also clear that the key-rate will be lowest when each $N^i$ attains this minimum value (any larger value of $N^i$ can only increase the key-rate of the actual protocol).

Given all this, we create a new EB protocol, denoted $\protEBM{2}$.  This protocol is identical to $\protEBM{1}$ except for the following changes:
(1) We change step 2 so that Eve prepares a state $\ket{\psi}_{AT^1\cdots T^p B}$ where, now, each register $A=R^0$, $L^i$, $R^i$, and $B=L^{p+1}$ consists of $\minN$ qubits exactly.
(2) Step 3 is changed so that each link $i$ simply agrees on a subset $\Theta^i$ (since both left and right parties on a link will always agree on the same subset for their measurements now).  However, to ensure the distribution of bases remains the same after ``discarding'' the rejected signals in $\protEBM{1}$, we take $\Theta^i \in \{0,1\}^{\minN}$ and the probability that $\Theta_j^i = 1$ is now $\px^2/(1-2\px(1-\px) - \beta)$.
  (3) Finally, Step 4 is removed since there are no longer any rejected qubits. \fullversion{Instead, Eve is preparing a smaller state simulating the worst case rejection strings.}{}

  There is one more modification we will make to simplify the security analysis.  Consider a particular link $i$ and basis choice $\Theta^i \in \{0,1\}^{\minN}$.  Let $m^i = wt(\Theta^i)$ and $n^i = \minN - m^i$ be the size of the $X$ and $Z$ basis measurement data on link $i$.  Let $\minX = \min_i m^i$ and $\minZ = \min_i n^i$.  Note that any measurement data larger than this value is simply discarded in a deterministic way by discarding any qubits after the cutoff point.  Making the same arguments as before, it is to Eve's benefit if these strings are all of equal size, but the smallest possible value.  We can use Hoeffding's inequality and add an additional abort case as we did when moving from  $\protEBM{1}$ to $\protEBM{2}$ to create a new protocol $\protEB$ (the actual protocol we'll analyze), where each link chooses a random measurement subset ensuring that the number of $X$ basis measurements is exactly $\minX$ in all links.  Of course, we must also ensure that the number of $Z$ basis measurements is $\minZ$ in all links - this can be done by further shrinking the total number of qubits Eve sends to all parties.  In particular, we use Hoeffding's bound to ensure, expect with probability $\epsilon_{\text{abort}}$, that:
  \begin{equation}\label{eq:minX}
    \minX = \minN\left(\frac{\px^2}{1-2\px(1-\px) - \beta} - \beta'\right)
  \end{equation}
and:
  \begin{equation}\label{eq:minZ}
    \minZ = \minN\left(1-\frac{\px^2}{1-2\px(1-\px) - \beta} - \beta'\right).
  \end{equation}
  \fullversion{
  Above:
  \begin{equation}
    \beta' = \sqrt{\frac{\ln\frac{2}{\epsilon_{\text{abort}}}}{2\minN}}.
  \end{equation}}
{Above, $\beta'$ is defined similarly to $\beta$, using $\minN$ instead of $N$.}
  
  Of course, since we are ensuring the number of one's in each $\Theta^i$ to be fixed at $\minX$, this is equivalent, now, to having each link $i$ choose a random subset $\Theta^i \subset \{1, 2, \cdots, \minX + \minZ\}$ of size $|\Theta^i| = \minX$.  This subset will index which qubits to measure in the $X$ basis, while any qubit not indexed by this subset will be measured in the $Z$ basis.  Of course, we also now assume that Eve creates an initial state where each party $L^i$ and $R^i$, now receives:
  \begin{equation}
    \minNN := \minX + \minZ = N(1-2\px(1-\px)-\beta)(1-2\beta')
  \end{equation}
  qubits.  Finally, we can reduce the protocol further by having all parties agree on a single subset.  In practice, each link will have it's own sampling subset $\Theta^i$.  However, having only a single subset chosen (say, Alice choosing a subset and sending it to everyone) can only benefit the adversary as there will be potentially less uncertainty for Eve; it can also easily be shown equivalent to the multi-subset case if all parties randomly permute their data.  Thus, we conclude with one final change to the protocol, namely only a single random subset is chosen of size $\minX$ and all parties measure this subset.

This is the final protocol we will actually analyze.  From our above discussion and analysis it is clear that the key-rate of $\protEB$ will serve as a lower-bound on the key-rate of protocol $\protEBM{0}$ (and, consequently, of the actual protocol $\protPM$).  The total failure probability of $\protPM$ will be, so far, at most $2(p+1)\epsilon_{abort}$.

\heading{Key-Rate Analysis}

  We now derive a bound on the key-rate of $\protEB$ (which will imply a lower bound on the key-rate of $\protPM$).  Our main result is described in the following theorem:

  \begin{theorem}
    Let $\epsilon > 0$ be given.  Let $\ket{\psi}_{R^0T^1T^2\cdots T^p L^{p+1}E}$ be the state Eve creates, where $T^i=L^iR^i$ and $L^i$, and $R^i$ consists of $\minNN$ qubits each.  Assume a subset $\Theta \subset \{1, \cdots, \minNN\}$ is chosen of size $\minX$ uniformly at random.  Each link $i$, consisting of registers $L^iR^{i+1}$, for $i = 0, \cdots, p$, will measure their qubits, indexed by $\Theta$, in the $X$ basis, producing outcomes $r^i$, and $l^{i+1}$.  Each STN broadcasts the parity of their measurement outputs, namely $q^i = l^i\oplus r^{i+1}$, for $i = 1, \cdots, p$.  Let:
    \begin{equation}
      q = r^0\oplus (l^1\oplus r^1) \oplus \cdots \oplus (l^p \oplus r^p) \oplus l^{p+1}.
    \end{equation}
    Ideally, if there is no noise, it should hold that $q$ is the zero string.

    After this, parties measure the remainder of their systems in the $Z$ basis.  Each STN will broadcast the parity of their $Z$ basis measurement results.  Let $P^i$ be the random variable determining $\stn{i}$'s parity broadcast for $Z$ basis states and $P = P^1\cdots P^p$.  Let $A_Z$ be the random variable determining Alice's $Z$ basis measurement of the remaining $R^0$ qubits.

    This entire experiment, conditioning on a particular subset $\Theta$ being chosen, and a particular $X$ basis outcome and broadcast of $\chi = r^0, q^1, \cdots, q^p, l^{p+1}$, can be modeled as a density operator $\rho_{AEP}(\Theta, \chi)$ (tracing out Bob and the STN's).  Then, except with probability at most $2\epsilon^{1/3}$, it holds that:
    \begin{equation}
      \Hmin^{4\epsilon + 2\epsilon^{1/3}}(A_Z|EP)_{\rho(\Theta, \chi)} \ge \minZ\left(1 - h\left(w(q) + \delta\right)\right)
    \end{equation}
    where the probability is over the subset choice and the observed $\chi$, and where:
    \begin{equation}\label{eq:delta}
      \delta = \sqrt{\frac{\minNN+2}{\minX\minNN}\ln\frac{2}{\epsilon^2}}.
    \end{equation}
  \end{theorem}
\begin{proof}
  Let $\ket{\psi}_{R^0T^1T^2\cdots T^p L^{p+1}E}$ be the state Eve creates.  It is not difficult to see that the sampling process used in $\protEB$ is the strategy $\samp{STN}$ discussed in Section \ref{sec:sampling} and analyzed in Lemma \ref{lemma:samp-stn}.  Using Theorem \ref{thm:sampling}, we can construct ideal states $\ket{\nu^\Theta}$ such that:
    $\ket{\nu^\Theta} \in \mathcal{G}_{STN,X}^\Theta,$ 
  where $\mathcal{G}_{STN,X}^\Theta$ is defined in Equation \ref{eq:good-STN} (it is the set of good words induced by $\samp{STN}$ using the $X$ basis in the spanning set definition) and, furthermore:
  \begin{equation}\label{eq:td}
    \frac{1}{2}\trd{\sum_{\Theta}P_T(\Theta)\kb{\Theta}\otimes\left(\kb{\psi} - \kb{\nu^\Theta}\right)} \le \sqrt{\epsilon^{cl}_{STN}} = \epsilon,
  \end{equation}
  where the last equality follows from our choice of $\delta$ and Lemma \ref{lemma:samp-stn}.

  We will analyze the ideal state, defined as $\sum_\Theta P_T(\Theta)\kb{\Theta}\otimes\kb{\nu^\Theta}$,  and compute the min entropy there.  Equation \ref{eq:td} and Lemma \ref{lemma:td-entropy} will allow us to promote the ideal state analysis to  the real state.

  Parties choosing a subset $\Theta$ is equivalent to measuring the subset register and observing a particular $\Theta$.  In the ideal state, this causes the system to collapse to $\ket{\nu^\Theta}$.  An $X$ basis measurement is performed on all qubits indexed by $\Theta$ (in each $L^i$ and $R^i$ register).  Each STN broadcasts the parity of their measurement result.  Let $q = r^0_\Theta \oplus (l^1_\Theta\oplus r^1_\Theta) \oplus \cdots \oplus (l^p_\Theta\oplus r^p_\Theta) \oplus l^{p+1}_\Theta$ be the result of XOR'ing all measurement results.  Since these are ideal states, by Equation \ref{eq:good-STN}, the post-measured state collapses to a state of the form:
  \begin{equation}\label{eq:pm-state}
    \ket{\nu^t_q} = \sum_{(r^0,\cdots, l^{p+1})\in J_q}\ket{r^0,l^1r^1, \cdots, l^pr^p, l^{p+1}}^X\ket{E_{r^0,\cdots,l^{p+1}}^{t,q}}
  \end{equation}
  where:
  \begin{align*}
    &J_q = \left\{(r^0, l^1r^1, \cdots, l^pr^p, l^{p+1})\in\Sigma_{\minZ}\st\right.\\
    &\left|w\left(r^0\oplus(l^1\oplus r^1)\oplus\cdots\oplus (l^p\oplus r^p)\oplus l^{p+1}\right) - w(q)\right| \le \delta\}.
  \end{align*}
  \fullversion{(See, also section \ref{sec:sampling} for more details on the quantum sampling framework we are using here.)}{}

  At this point, parties will measure their remaining qubits in the $Z$ basis, and each STN will broadcast the parity of their $Z$ basis measurement results.  \fullversion{Bob will take these broadcasts and XOR to his $Z$ basis measurement result, yielding his raw key; Alice's raw key is simply her direct measurement result.}{}  We are interested in computing a bound on the quantum min entropy of Alice's measurement result, given Eve's system and all the parity broadcasts.

  Let's consider a single STN: instead of measuring immediately in the $Z$ basis and broadcasting the result, we can equivalently assume each STN will apply a double CNOT to their $L^i$ and $R^i$ registers, XORing their results (in the computational basis) into a ``blank'' ancilla.  Then, the STN will measure this ancilla to produce the parity message.
  
  More specifically, consider $\stn{i}$ and qubit $j$ (out of $\minZ$).  Namely, we are considering the $j$'th qubits in both registers $L^i$ and $R^i$.  Ordinarily, the STN will measure this system in the $Z$ basis, XOR the results classically, and broadcast that bit.  However, instead, we may consider delayed measurements: the STN may equivalently prepare a blank ancilla in a $\ket{0}$ state, apply a CNOT operation using the $j$'th qubit in $L^i$ as the control and the new ancilla as the target, followed by a second CNOT, this time using the $j$'th qubit in $R^i$ as the control and, again, the same ancilla as target.  Thus, it will map $\ket{x,y}_{L^i_jR^i_j}\ket{0}_{P^i_j}$ to $\ket{x,y}_{L^i_jR^i_j}\ket{x\oplus y}_{P^i_j}$, where $x$ and $y$ are single bits (note this definition is with respect to the computational, $Z$ basis).  Measuring the ancilla at this point and then later measuring the $L^i$ and $R^i$ registers in the $Z$ basis, will produce the same system as if $\stn{i}$ had simply measured the $L^i$ and $R^i$ registers in the $Z$ basis and computed the XOR classically.

  Given the action of this unitary operation on $Z$ basis states, namely $\ket{x,y}\ket{0} \mapsto \ket{x,y}\ket{x\oplus y}$, its action on $X$ basis states (which is what Equation \ref{eq:pm-state} is written in), is found to be $\ket{a,b}_{L^i_jR^i_j}^X\ket{0}_{P^i_j}^Z= \frac{1}{2}(\ket{00} + (-1)^a\ket{01} + (-1)^b\ket{10} + (-1)^{a\oplus b}\ket{11})\ket{0}$ which maps to:
  \begin{align*}
                         &\frac{1}{2}(\ket{00} + (-1)^{a\oplus b}\ket{11})\ket{0} + (-1)^a \frac{1}{2}(\ket{01} + (-1)^{a\oplus b}\ket{10})\\
    & = \frac{1}{\sqrt{2}}\ket{\phi_0^{a\oplus b}}_{L^i_jR^i_j} \ket{0}_{P^i_j} + \frac{(-1)^a}{\sqrt{2}}\ket{\phi_1^{a\oplus b}}_{L^i_jR^i_j}\ket{1}_{P^i_j},
  \end{align*}
  where $\ket{\phi_x^y} = \frac{1}{\sqrt{2}}(\ket{0,x} + (-1)^y\ket{1, 1\oplus x})$.  Above, we are denoting this new register as $P^i$ since it will store $\stn{i}$'s parity broadcast.

  Of course, the above map is applied to all $\minZ$ qubits; the action on such a basis state is easily seen to be:
  \begin{equation}
    \ket{l^i, r^i}_{L^iR^i}\ket{0}_{P^i} \mapsto \sum_{c^i\in\{0,1\}^\minZ}\frac{(-1)^{c^i\cdot l^i}}{\sqrt{2^\minZ}}\ket{c^i}_{P^i} \ket{\phi_{c^i}^{l^i\oplus r^i}}_{L^iR^i},
  \end{equation}
  where, above, we permuted the $P^i$ and $L^iR^i$ registers only for clarity in our subsequent presentation and where $c^i\cdot l^i$ is the bit-wise modulo two dot product, namely $c^i\cdot l^i = c^i_1l^i_1\oplus\cdots\oplus c^i_\minZ l^i_\minZ$.  Furthermore, by $\ket{\phi_{c^i}^{l^i\oplus r^i}}_{L^iR^i}$, we mean $\ket{\phi_{c^i_1}^{l^i_1\oplus r^i_1}} \otimes\ket{\phi_{c^i_2}^{l^i_2\oplus r^i_2}}\otimes\cdots$

  All STNs apply this delayed measurement map; due to linearity, the joint system $\ket{\nu_q^t}$ (Equation \ref{eq:pm-state}) evolves to a state we denote $\ket{\zeta_q^t}$ which is found to be:
  \begin{align}
    &\ket{\zeta_q^t} = \frac{1}{\sqrt{2^{\minZ\cdot p}}}\sum_{c^1,\cdots, c^p\in\{0,1\}^\minZ}\ket{c^1\cdots c^p}_P\notag\\
                                                       &\otimes\sum_{(r^0,\cdots, l^{p+1})\in J_q}(-1)^{c\cdot l}\ket{r^0}^X\ket{\phi_{m^1}^{l^1\oplus r^1}}\cdots\ket{\phi_{c^p}^{l^p\oplus r^p}}\ket{l^{p+1}}^X\notag\\
    &\otimes\ket{E^{t,q}_{r^0,\cdots, l^{p+1}}},
  \end{align}
  where $c\cdot l = c^1\cdot l^1 + \cdots + c^p\cdot l^p$.
  
  At this point, the STN's will measure their respective $P$ registers and broadcast the message result (the message being the parity of their measurements or, in this case, the parity of what their measurements will eventually be since we are working with a delayed measurement setup now).  This cause the state to collapse to the mixed state $\sum_c\kb{c}\otimes\kb{\zeta^t_{q,c}}$, where $\zeta^t_{q,c} = \sum_{(r^0,\cdots, l^{p+1})\in J_q}(-1)^{c\cdot l}\ket{r^0}^X \ket{\phi_{c^1}^{l^1\oplus r^1}}\cdots\ket{\phi_{c^p}^{l^p\oplus r^p}}\ket{l^{p+1}}^X$ $\otimes\ket{E^{t,q}_{r^0,\cdots, l^{p+1}}}$
where the sum over $c$ is actually over $c=(c^1, \cdots, c^p)$, where each $c^i \in \{0,1\}^\minZ$.  Note we are disregarding the normalization term which may be absorbed into Eve's vectors.

  Let's consider a particular parity broadcast $c$ and the post measured state $\ket{\zeta^t_{q,c}}$ defined in the equation above.  We may re-write these states in the following form $\ket{\zeta^t_{q,c}} \cong$
  \begin{align}
    &\sum_{l^1,r^1,\cdots, l^{p+1} \in \{0,1\}^\minZ} (-1)^{c\cdot l}\ket{\phi_{c^1}^{l^1\oplus r^1}}\cdots \ket{\phi_{c^p}^{l^p\oplus r^p}}\ket{l^{p+1}}^X\notag\\
    &\otimes\sum_{r^0\in J_{q}(l^1\oplus r^1,\cdots l^p\oplus r^p, l^{p+1})}\ket{r^0}^X\ket{E^{t,q}_{r^0,\cdots, l^{p+1}}}.\notag\\
                        &=\sum_{x^1, x^2, \cdots, x^p, l^{p+1}\in\{0,1\}^{\minZ}}\ket{\phi_{c^1}^{x^1}}\cdots\ket{\phi_{c^p}^{x^p}}\ket{l^{p+1}}^X\notag\\
    &\otimes\sum_{r^0 \in J_q(x^1, \cdots, x^p,l^{p+1})}\ket{r^0}^X\ket{F^{t,q}(c, r^0, x^1,\cdots, x^p, l^{p+1})},
  \end{align}
  where $J_q(x^1,\cdots, x^p, l^{p+1}) =$
  \begin{equation}
     \{r^0\in\{0,1\}^\minZ \st |w(r^0\oplus x^1\oplus\cdots\oplus x^p\oplus l^{p+1}) - w(q)|\le \delta\}
  \end{equation}
  and $\ket{F^{t,q}(c, r^0, x^1,\cdots, x^p, l^{p+1})} =$
  \begin{equation}
     \sum_{\substack{l^1, r^1\in\{0,1\}^\minZ\\ \st l^1\oplus r^1 = x^1}}\cdots \sum_{\substack{l^p, r^p\in\{0,1\}^\minZ\\ \st l^p\oplus r^p = x^p}} (-1)^{c\cdot l}\ket{E^{t,q}_{r^0, l^1, \cdots, l^{p+1}}}.
  \end{equation}
  
  Now, returning to the general mixed state $\sum_c\kb{c}\otimes\kb{\zeta^t_{q,c}}$, each STN will measure their $L^i$ and $R^i$ systems in the $Z$ basis and Bob will measure his register (the $L^{p+1}$ register) in the $Z$ basis.  Since we care only about Alice's system at this point, we will then discard the system.  Of course, this is mathematically equivalent to simply tracing out these systems from $\ket{\zeta^t_{q,c}}$ immediately.  Doing so leads the mixed state:
  \begin{equation}
    \sum_c\kb{c}\otimes\sum_{x^1,\cdots, l^{p+1}}P\left(\smashoperator[r]{\sum_{r^0 \in J_q(x^1, \cdots, x^p,l^{p+1})}}\ket{r^0}^X\ket{F^{t,q}(c, r^0, x^1,\cdots)}\right)
  \end{equation}
where $P(\ket{z}) = \kb{z}$
  At this point a measurement of Alice's register ($R^0$) is made. Equation \ref{eq:mixed-ent}, along with Lemma \ref{lemma:superpos}, can be used to show:
  \begin{equation}
    \Hmin(A|EP) \ge \min_{c, x^1,\cdots, x^p,l^{p+1}}(\minZ - \log_2|J_q(x^1,\cdots, x^p, l^{p+1})|).
  \end{equation}
  It is not difficult to show that:
  \begin{align*}
    &|J_q(x^1, \cdots, x^p, l^{p+1})|\\
    &\le |\{i\in\{0,1\}^\minZ \st w(i) \le w(q) + \delta\}|\le 2^{\minZ h(w(q) + \delta)},
  \end{align*}
  where the last inequality follows from the well-known bound on the volume of a Hamming ball.

   This completes the analysis of the ideal state.  Thanks to Equation \ref{eq:td}, this ideal state is $\epsilon$-close to the real one; Lemma \ref{lemma:td-entropy}, then allows us to complete the proof (taking the random variable $X$ in that lemma to be the subset chosen and the observed $q$).
 \end{proof}

 The above gives us a bound, with high probability, on the quantum min entropy of Alice's raw key conditioned on Eve's side information, and also conditioning on a particular run of the protocol (i.e., conditioning on an actual $X$ basis observation being made).  Using Equation \ref{eq:keyrate}, this leads us directly to a key-rate expression for an STN chain.  In particular, let $\epsilon_{PA} = 9\epsilon+4\epsilon^{1/2}$, then except with probability at most $\epsilon_{fail} = 2\epsilon^{1/3} + 2(p+1)\epsilon$ (where the last term is due to the abort conditions in the event subsets are too small as discussed earlier in our reductions), the final secret key size will be:
 \begin{equation}\label{eq:key-rate-STN}
   \ell_{STN} = \minZ\left(1 - h(w(Q) + \delta)\right) - \leakEC - 2\log\frac{1}{\epsilon}
   \end{equation}
   where $\leakEC$ is the error correction leakage and $\minZ$ and $\minX$ can be found on Equations \ref{eq:minZ} and \ref{eq:minX}.


\section{Evaluations}

Now that we have a finite-key bound for the STN chain, we can evaluate.  While our key-rate proof applies to any noise scenario, will evaluate assuming each link in the chain is a depolarization channel with parameter $Q$.  In this case, the $Z$ or $X$ basis noise in each individual link is simply $Q$ (which we call the \emph{link-level noise}).  Of course, an STN network cannot determine the link-level noise, since no sampling is done at the link level.  Instead, we need to determine the expected value of $w(q)$, where $q$ is the ``additive'' error in each link.  Namely, we need to determine the probability of an error between Alice and Bob after each STN transmit their parity bits.

\fullversion{
  It is not difficult to see in a chain with $p$ STN's (thus $p+1$ total links), an error can only occur if there are an odd number of errors in the total chain.  For instance, in a chain with three links, if there is an error in one link but not two, there will be an error in the entire chain.  However, if there is an error in two of the links, those errors will ``cancel out'' when the parity measurements are transmitted and XOR'd together.  Thus, it is not difficult to see that the expected value of $w(q)$ is simply:
\begin{equation}\label{eq:total-noise}
  w(q) = \sum_{i=0}^{\left\lceil\frac{p+1}{2}\right\rceil - 1}{p+1 \choose 2i+1}Q^{2i+1}(1-Q)^{p-2i}
\end{equation}
}
{
    It is not difficult to see in a chain with $p$ STN's (thus $p+1$ total links), an error can only occur if there are an odd number of errors in the total chain. Thus, the expected value of $w(q)$ is simply:
\begin{equation}\label{eq:total-noise}
  w(q) = \sum_{i=0}^{\left\lceil\frac{p+1}{2}\right\rceil - 1}{p+1 \choose 2i+1}Q^{2i+1}(1-Q)^{p-2i}
\end{equation}

  }
This allows us to evaluate our key-rate equation as derived in Equation \ref{eq:key-rate-STN}.  Key-rates are compared with a chain using regular trusted nodes (denoted simply ``TN'' where, recall, such trusted nodes perform a full QKD stack of sampling, error correction, and privacy amplification).  For a TN chain with $p$ regular TNs, we simply use the standard BB84 finite key rate equation from \cite{tomamichel2012tight}, namely:
\begin{equation}\label{eq:rate-TN}
  \ell_{BB84} = \ell_{TN} = \minZ(1 - h(Q + \mu)) - \leakEC - 2\log\frac{2}{\epsilon'}
\end{equation}
where, note, above the entropy depends on the link level noise $Q$ and not the total noise $w(q)$.  Above, we have:
 $ \mu = \sqrt{\frac{\minZ + \minX}{\minZ\minX}\frac{\minX+1}{\minX}\ln\frac{2}{\epsilon'}}.$
For our evaluations, we set $\epsilon = 10^{-30}$, $\epsilon_{abort} = 10^{-10}$, and $\epsilon' = 10^{-10}$.  This provides an error and failure probability on the order of $10^{-10}$ for both our STN result and the above TN result.  We set $\leakEC = h(w(q)+\delta)$ for the STN case, and $\leakEC = h(Q + \mu)$ for the TN case.

\begin{figure}
    \centering
    \includegraphics[width=.48\linewidth]{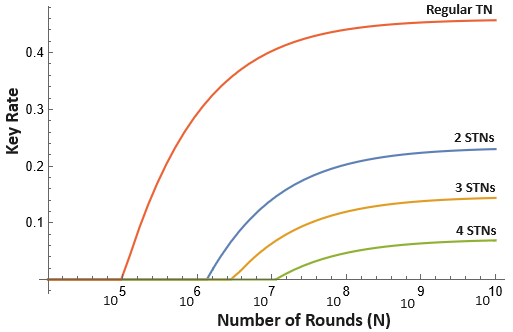}
    \includegraphics[width=.48\linewidth]{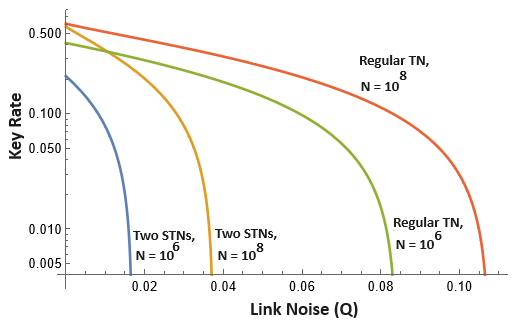}
    \caption{Left: Comparing the finite key-rates of an STN chain (bottom three: blue, yellow, and green), with a regular TN chain (top: red) as the total number of signals, $N$, increases.  Here, we set the link level noise to be $Q = 2\%$ and $\px = 0.2$ for all tests.  Note that as the number of STNs in the chain increases while the link-level noise remains constant, the total key-rate degrades.  This is known to happen asymptotically as shown in \cite{STN}.  Regular TN networks are limited only by the link level noise and so the number of trusted nodes is irrelevant in this case.  Right: Comparing a regular TN chain with an STN chain consisting of two STNs as the link level noise $Q$ increases.  Here we compare $N=10^6$ and $N=10^8$ total signals.  We set $\px = .2$ as before.}
    \label{fig:STNvsTN}
  \end{figure}

  \begin{figure}
    \centering
    \includegraphics[width=.48\linewidth]{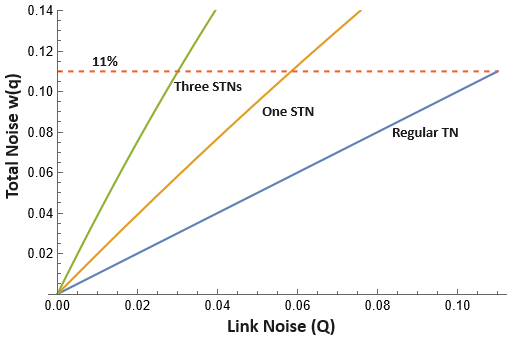}
    \includegraphics[width=.48\linewidth]{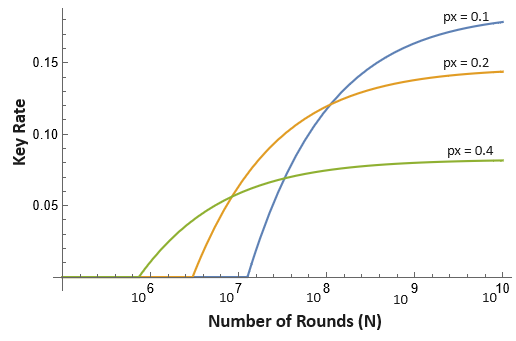}
    \caption{Left: Showing how the total noise (Equation \ref{eq:total-noise}) increases as the link noise ($Q$) increases.  For a regular TN, the total noise depends only on a single link's noise level; as the number of STN's increases, the total noise increases drastically.  Once the total noise surpasses $11\%$, it is impossible for a key to be distilled given our key-rate expression (or the asymptotic rate from \cite{STN}).
    Right: Evaluating the finite key-rates of an STN chain with three STNs for a fixed link level noise of $Q = 2\%$ but varying $\px$.}
    \label{fig:total-noise-and-px}
\end{figure}

Figure \ref{fig:STNvsTN} shows a comparison in key-rates between an STN chain and a TN chain.  Note that the noise tolerance of an STN network is significantly lower than a regular trusted node network.  However, looking at Equation \ref{eq:total-noise}, this is not surprising; indeed as the link-level noise increases, the total noise between Alice and Bob in an STN chain may increase dramatically, as shown in Figure \ref{fig:total-noise-and-px} (Left).

Furthermore, this decrease in key-rate as the number of STNs increases is not unique to our proof and was discovered, at least in the asymptotic case, in \cite{STN}.  Of course, the finite key results cannot be better than asymptotic results.  Note that we are the first to derive a finite key security proof for an STN chain, so we cannot compare the finite key results to other work in STN chains.

Of course, in finite key settings, multiple parameters affect performance.  In addition to the total number of signals sent, the value of $\px$ will also greatly affect key-rates.  This is shown in Figure \ref{fig:total-noise-and-px} (Right).  Note that for small values of $\px$, higher key-rates are possible for larger $N$, however for larger values of $\px$, the overall key-rate will be lower, but one will attain a positive key-rate for smaller $N$.


\fullversion{\subsection{Cost Comparison}}{\textbf{Cost Comparison: }}
Despite the fact that STN chains provide lower noise tolerances, there are still potential benefits to using STN networks if the noise is ``low enough.''  In particular, since each STN does not need to run EC and PA every time a key is derived for Alice and Bob, there may be cost savings in running an STN network.  To formally argue this, we derive a novel cost function for a QKD chain consisting of STNs or TNs.  Our cost function will take into account the cost of running EC and PA; to be fair, it must also take into account the fact that an STN chain, though not always required to perform such operations, will occasionally need to do so, to replenish their secret key pools for the authenticated channel.

Let's consider the cost of running a TN first.  Alice and Bob wish to use the TN chain to establish a shared secret key.  To do, so $N$ qubits are transmitted pair-wise, leading to a secret key of size $\ell_{TN} = \ell_{TN}(N, Q)$, where $\ell_{TN}$ is from Equation \ref{eq:rate-TN} and we use $\ell_{TN}(N,Q)$ to show it's dependence on $N$ and the link noise $Q$ (the additional $\epsilon$ factors do not contribute significantly for large $N$ and so we do not explicitly write them out, though they do appear in our evaluation of $\ell_{TN}$ of course).  To produce this key, Alice and Bob both run  EC and  PA.  Furthermore, to produce this key, each pair of TN's must run EC and PA \emph{twice} (one with their neighbor to the left and one with their neighbor to the right).  We will use $EC(N,Q)$ to be the cost of running these EC and PA processes when the total number of signals sent was $N$ and with a noise in the raw key of $Q$.  We will assume that some of this key is used to replenish each TN's pre-shared key for authentication and so they do not need to do any further computation beyond this.  In this case, the cost of running a TN chain is:
\begin{equation}\label{eq:cost-TN}
  \mathcal{C}_{TN} = \frac{\text{cost}}{\text{secret key bits}} = \frac{(2p+2)EC(N,q)}{\ell_{TN}(N,Q)}
\end{equation}

For the STN the case is more involved.  When Alice and Bob want to establish a secret key, they will send $N$ qubits through the chain.  Then, only Alice and Bob will run EC and PA, leading to a secret key size of $\ell_{STN} = \ell_{STN}(N,w(q), p)$, where $\ell_{STN}$ is from Equation \ref{eq:key-rate-STN} (note the additional dependence on $p$).  The STN's do not need to perform EC and PA for this key; however they did use up some of their shared secret key pool for their authenticated classical communication (see Figure \ref{fig:STN-chain}).  This key pool cannot be refreshed immediately as it could with the TN case, since the STNs did not perform a full QKD operation (they did not perform EC and PA).  This key-pool will need to be refreshed sometime.

Let's assume that each STN starts with $k$ secret key bits for their authenticated communication.  Let's also assume that for Alice and Bob to establish a secret key using $N$ rounds of the STN chain, this will require $c(N)$-bits to be used from the secret key pool of each STN \fullversion{(this number does not depend on the noise of the channel, since the communication cost depends only on the number of rounds used, $N$)}{}.  After Alice and Bob use the STN network $J$ times (each time establishing a secret key of size $\ell_{STN}(N, w(q), p)$), each STN has a secret key pool of size $k - Jc(N)$.  Once this is ``low enough'', each STN must, independently, run pairwise QKD with their neighbors, sending $N$ rounds of qubits, and performing EC and PA with each neighbor.  After this, each STN will now have an additional $\ell_{BB84}(N, Q)$ key bits in their secret key pools for authentication.  We will assume that the STN's will perform this pair-wise QKD whenever they have $c(N)$ bits remaining in their secret key pools and, so, they must do this after the $J = (k-c(N))/c(N)$'th round.

Summarizing, the STN's do not need to perform any EC or PA for $J$ key establishments of the network.  During these $J$ rounds, Alice and Bob have established $J\times\ell_{STN}(N, w(q), p)$ secret key bits; of course these users must be performing EC and PA for each of their $J$ secret keys.  Finally, only after the $J$'th key is established do the STN's need to perform their own QKD establishment with their adjacent neighbors.  This will provide them with additional key bits for their pool based on the link-level noise.  Note that Alice and Bob must also do this to refresh their shared keys with their neighboring STN's.  This leads to a final cost function of:
\begin{align}\label{eq:cost-STN}
  \mathcal{C}_{STN} & = \frac{2J \times EC(N, w(q)) + (2p+2)EC(N, Q)}{J\times \ell_{STN}(N, w(q), p)}
\end{align}
If we assume $k = \ell_{BB84}(N,Q) = \ell_{TN}(N, Q)$, then
\fullversion{
  \begin{equation}
    J = \frac{\ell_{BB84}(N,Q) - c(N)}{c(N)}
  \end{equation}
}
{
  $J = (\ell_{BB84}(N,Q) - c(N))/c(N)$
  }

To evaluate and compare, we set $c(N) = \log_2N$, since information theoretic authentication generally requires a logarithmic number of secret keys \cite{wegman1981new}.  We also set $EC(N, Q) = N$, that is, we will simply assume the cost of EC and PA are linear in the number of rounds. Of course other scenarios may be evaluated.  Note that the ``cost,'' as we evaluate it, is a unit-less function in our case: it may be related to running time, memory usage, etc.  Users of a STN/TN network should modify this to suite their needs.  Our results are shown in \fullversion{Figures \ref{fig:cost-sig} and \ref{fig:cost-noise}.}{Figure \ref{fig:cost-sig}.}  It is clear from these figures that STN chains may be much more cost effective in low-noise scenarios; however, in high noise scenarios (i.e., high link-level noise), regular TNs may be more cost effective.

\fullversion{
\begin{figure}
    \centering
    \includegraphics[width=.48\linewidth]{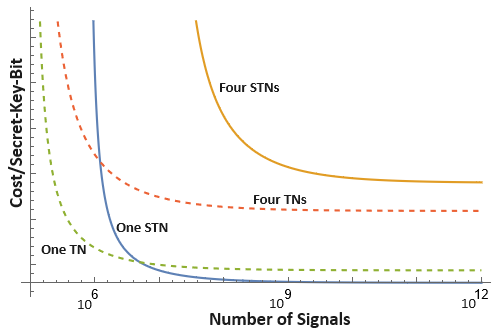}
    \includegraphics[width=.48\linewidth]{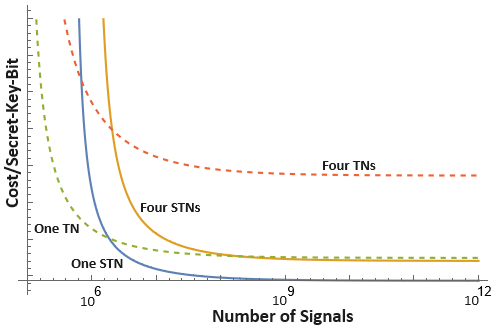}
    \caption{Comparing the cost per secret key bit of STN chains (Solid Lines, Equation \ref{eq:cost-STN}) and regular TN chains (Dashed Lines, Equation \ref{eq:cost-TN}) as the number of signals per key establishment round ($N$) increases.  Left: Link level noise $Q = 2\%$; Right: Link level noise $Q = 1\%$.  In both, we have $\px = 0.2$.  Note that STNs are more cost effective, according to our cost function above, for lower levels of noise than the comparably sized TN chain.}
    \label{fig:cost-sig}
  \end{figure}

  \begin{figure}
    \centering
    \includegraphics[width=.48\linewidth]{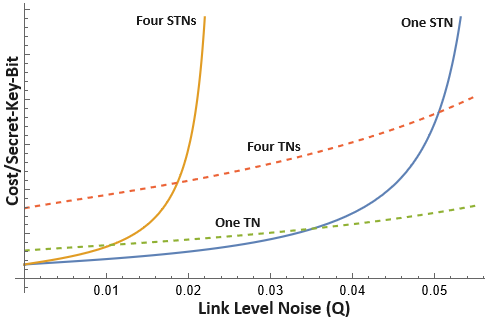}
    \caption{Comparing the cost per secret key bit of STN chains (Solid Lines) and regular TN chains (Dashed Lines) as the link level noise $Q$ increases.  Here, we set the number of signals for each key establishment round to be $N = 10^{10}$ and $\px = 0.2$.  These figures again demonstrate that STNs may be more cost effective, for lower levels of noise, than the comparably sized TN chain.}
    \label{fig:cost-noise}
  \end{figure}
}
{
  \begin{figure}
    \centering
    \includegraphics[width=.48\linewidth]{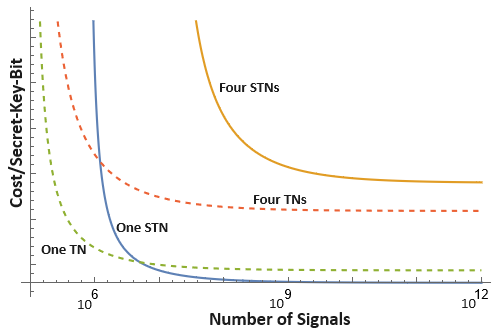}
    \includegraphics[width=.48\linewidth]{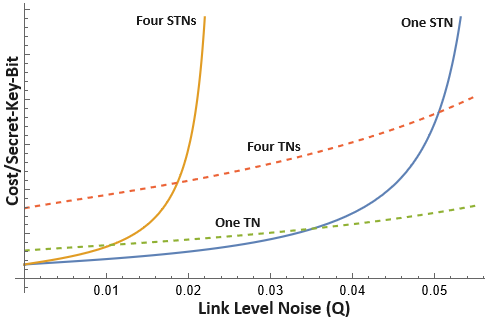}
    \caption{Left: Comparing the cost per secret key bit of STN chains (Solid Lines, Equation \ref{eq:cost-STN}) and regular TN chains (Dashed Lines, Equation \ref{eq:cost-TN}) as the number of signals per key establishment round ($N$) increases.  Link level noise $Q = 2\%$ and $\px = 0.2$. Right: Same comparison, but as link-level noise increases.  Here, we set $N = 10^{10}$.}
    \label{fig:cost-sig}
  \end{figure}

  }






\section{Closing Remarks}

In this paper, we derived a new proof of security for an STN chain in the finite key setting.  To our knowledge, this is the first time a finite-key security proof has been achieved for an STN chain.  \fullversion{Our proof methods may have broad application to other QKD networking scenarios.}{}  We also evaluate the STN network performance in a variety of scenarios and compare with a regular TN network.  Finally, we derive a new cost function to more effectively compare STN and TN networks.

In general, STNs have lower noise tolerances, however they may be more cost effective in some scenarios.  Since STNs do not need to perform error correction and privacy amplification every time end-users want to establish a secret key, they can be equipped with slower computational hardware.  \fullversion{Our cost function demonstrates that for low levels of noise STNs can be much more cost effective in the long run, when compared to regular TN networks.  There may also be security benefits to STN chains as explained in \cite{huang2022stream}.}{}  Overall, our work in deriving a new finite-key proof of security for STNs can be beneficial to further research into developing a cost-effective QKD network.

Many interesting future problems remain.  Dealing with channel loss and imperfect sources would be interesting.  We suspect our proof methods can be suitably adapted to handle this case, perhaps combined with decoy state methods \cite{hwang2003quantum,lo2005decoy,wang2005beating}, though a full proof we leave as future work.


\balance


\end{document}